\documentclass{article}
\usepackage[letterpaper, portrait, margin=1in]{geometry}
\usepackage{graphicx, caption}
\usepackage{amsfonts}
\usepackage{amsmath,bm}
\usepackage{cite}
\usepackage{IEEEtrantools}
\usepackage{amssymb}
\usepackage{amsthm}
\usepackage{authblk}
\usepackage{multirow}
\usepackage{array}

\theoremstyle{plain}
\newtheorem{thm}{Theorem}[section] 

\newtheorem{defn}[thm]{Definition} 

\newtheorem{lemma}[thm]{Lemma}

\newtheorem{conjecture}[thm]{Conjecture}
\newtheorem{condition}[thm]{Condition}

\newcommand\blfootnote[1]{%
	\begingroup
	\renewcommand\thefootnote{}\footnote{#1}%
	\addtocounter{footnote}{-1}%
	\endgroup
}

\begin{document}

\title{On the depth overhead incurred when running quantum algorithms on near-term quantum computers with limited qubit connectivity}
\date{}
\author[1,2]{Steven Herbert}
\affil[1]{Department of Applied Mathematics and Theoretical Physics, University of Cambridge, UK}
\affil[2]{Cambridge  Quantum  Computing  Ltd,  9a  Bridge  Street,  Cambridge,  CB2  1UB,  UK}
\setcounter{Maxaffil}{0}
\renewcommand\Affilfont{\itshape\small}

\captionsetup[figure]{labelfont={bf},name={Fig.},labelsep=period}

\captionsetup[table]{labelfont={bf},name={Table},labelsep=period}

\maketitle

\begin{abstract}
\noindent This\blfootnote{sjh227@cam.ac.uk} paper addresses the problem of finding the depth overhead that will be incurred when running quantum circuits on near-term quantum computers. Specifically, it is envisaged that near-term quantum computers will have low qubit connectivity: each qubit will only be able to interact with a subset of the other qubits, a reality typically represented by a qubit interaction graph in which a vertex represents a qubit and an edge represents a possible direct 2-qubit interaction (gate). Thus the depth overhead is unavoidably incurred by introducing \textit{swap} gates into the quantum circuit to enable general qubit interactions. This paper proves that there exist quantum circuits where a depth overhead in $\Omega(\log n)$ must necessarily be incurred when running quantum circuits with $n$ qubits on quantum computers whose qubit interaction graph has finite degree, but that such a logarithmic depth overhead is achievable. The latter is shown by the construction of a 4-regular qubit interaction graph and associated compilation algorithm that can execute any quantum circuit with only a logarithmic depth overhead.
\end{abstract}

\section{Introduction}
\label{intro}
It has been shown that the quantum circuit model (also known as the quantum gate model), in which the component qubits undergo single and 2-qubit gates from a universal set, is sufficient to express general quantum algorithms \cite{yao93,Deutsch73,elgates}.  As quantum computers move from theory to reality, it is becoming apparent that the envisaged physical limitations on the connectivity of the component qubits in near-term quantum computers presents a significant challenge to the implementation of quantum algorithms. So it follows that finding a method to efficiently run the quantum circuits that embody these algorithms on near-term quantum computers with limited connectivity is one of the foremost priorities of the designers thereof.\\
\indent The connectivity of quantum computers can be represented as a simple graph (termed the \textit{qubit interaction graph}), with each vertex representing a single qubit, and each edge corresponding to a possible 2-qubit interaction, i.e., the quantum computer can perform 2-qubit gates on `connected' (adjacent in the interaction graph) qubits (see Fig.~\ref{f0}). It follows that the physical limitations on the connectivity of the component qubits is manifested as the qubit interaction graph having low degree, for example the architecture of Google's quantum computer is a rectangular grid of connected qubits (degree 4) \cite{Google} as is that of IBM \cite[Fig.~4]{IBM}, whilst Rigetti's proposed 19-qubit architecture has maximum degree 3 \cite[Fig.~1]{rigetti}. These quantum computers all consist of planes of superconducting qubits, and thus the average degree and possible connectivity of the qubit interaction graph is necessarily restricted by this planarity. An alternative concept, which need not suffer from this restriction and which is beginning to gain attention in the field, is the idea of a \textit{networked quantum computer}, in which nodes such as ion-traps \cite{Nickerson,Nigmatullin} or Nitrogen-Vacancy centres \cite{delft} are connected by photonic links. In near-term networked quantum computers, it is expected that each node will have a single qubit available for computation (e.g., \cite{Nickerson,Nigmatullin}), and thus the qubit interaction graph exactly corresponds to the network of photonic links, which need not be planar (or even embeddable in $k$-dimensional space), but is still not expected to be completely connected.\\ 
\begin{figure}[!h]
	\centering
	\includegraphics[width=0.9\textwidth]{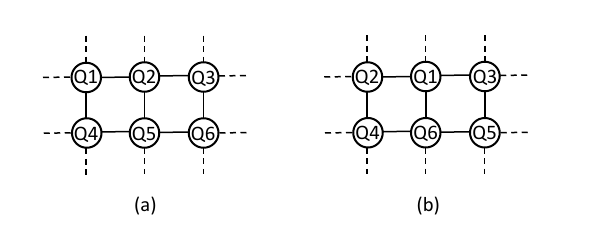}
	\captionsetup{width=.9\linewidth}
	\caption{\small{(a) Part of the interaction graph of quantum computer: in this case, a rectangular grid of qubits. Qubit Q1 cannot interact directly with qubit Q6 -- such an interaction can only be achieved by first performing a swap between Q1 and Q2 and also between Q5 and Q6 (these two swaps can be executed in parallel, as they are along different edges of the graph) so that Q1 and Q6 are adjacent, as shown in (b).}}
	\label{f0}
\end{figure}
\begin{figure}[!h]
	\centering
	\includegraphics[width=0.9\textwidth]{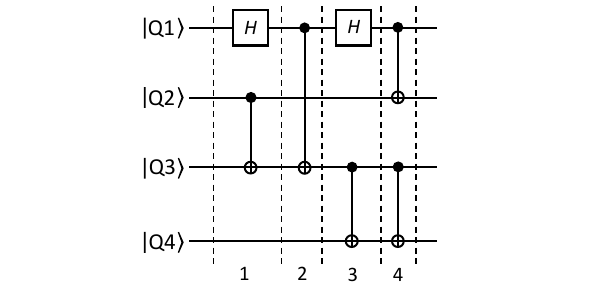}
	\captionsetup{width=.9\linewidth}
	\caption{\small{ The beginning of a quantum circuit decomposed into four layers (with single qubit gates included for illustration). Note that decomposing a quantum circuit into layers implies that the gates have been most efficiently packed into layers.}}
	\label{f05}
\end{figure}
\indent Notwithstanding the issue of whether a particular qubit interaction graph need be embeddable in $k$-dimensional space, the fact that \textit{all} of these proposed quantum computer architectures have limited qubit connectivity conflicts with the quantum circuit model which allows general 2-qubit interactions and thus implicitly assumes a completely connected graph. To remedy this, Beals \textit{et al.} demonstrate the use of \textit{swap} gates (hereafter usually referred to simply as `swaps' for brevity), which can be inserted into the quantum circuit to enable it to be executed\footnote{Throughout this paper, the term `execute' is used to refer to a quantum circuit being run to completion.} on a quantum computer with limited qubit connectivity \cite{Beals}. Therefore the central question is how to insert these swaps such that the circuit will run on a quantum computer with limited connectivity, whilst minimising the adverse effect on the algorithm's performance (e.g., its run-time). To answer this question, it is often helpful to consider the quantum circuit to be decomposed into successive layers of disjoint interactions (see Fig.~\ref{f05}). In the literature, the typical approach is to use swaps to permute the qubits such that an arrangement is reached in which all of the 2-qubit gates required at the current layer can be performed on connected pairs of qubits (and repeat until the circuit is completed). Thus the question of how to run quantum circuits on near-term quantum computers with limited qubit connectivity is typically reduced to the question of how to emulate a completely connected graph on a finite-degree regular graph\footnote{The definition of emulating the completely connected graph is given in Definitions~\ref{defref1} and \ref{defref2} in Section~\ref{def}, where it is also highlighted that this is actually not, in general sufficient (i.e., Condition~\ref{condref} is also necessary).}. In particular, architectures where the graph is embedded in a finite number of spatial dimensions (a more restrictive condition than the fixed degree of the graph considered herein) has been widely addressed, for example by Fowler, Devitt and Hollenberg \cite{fow}, Maslov \cite{mas} and Kutin \cite{kut} for one dimension and Pham and Svore \cite{pham} for two dimensions. More generally, Cheung, Maslov and Severini present the result that a graph embedded in $k$ dimensional space can emulate the completely connected graph in depth $\mathcal{O}(\sqrt[k]{n})$ for an $n$-qubit device \cite{che}.\\
%
\indent Whilst these results are important for understanding the general capability of near-term quantum computers to execute quantum algorithms, as already noted the qubit interaction graph need not be related to $k$-dimensional space for networked quantum computers. Moreover, there is no requirement that the algorithm need be compiled onto the quantum computer by modifying the quantum circuit such that it consists of successively repeating swapping and gating blocks (as has been assumed in these articles); in general swapping to re-order the qubits ahead of future gates may occur concurrently with present gates. Therefore a more general approach to bounding the performance of limited connectivity quantum computers, accounting for the possibility of performing gates at any time: This is the central question addressed in this paper.\\
\indent The analysis in this paper is restricted to noiseless synchronous machines, i.e., where each qubit pair simultaneously undergoes a gate or a swap in unit time -- or, in other words, a layer takes unit time (note that qubits can only interact with one neighbour at a time). Initially only full-width quantum circuits are considered, i.e., no qubits are reserved to be ancillas to aid the execution of the quantum circuit, instead all qubits are \textit{application qubits}\footnote{The term `application qubit' is used for consistency with the nomenclature used by Nigmatullin \textit{et al} \cite{Nigmatullin} and, moreover, is preferred to `logical qubit' to avoid unintended connotations of error-correction.} available to be used in the quantum algorithm itself. As well as providing a simple starting point, such a setting is physically reasonable for near-term quantum computers, and moreover, is consistent with that considered in the aforementioned literature \cite{fow, mas, kut, pham} and thus constitutes a useful benchmark to gauge the impact that limited qubit connectivity will have on the overall performance of quantum algorithms. Subsequently, in Section~\ref{general}, the possibility of using ancillas to aid the process of running quantum algorithms on quantum computers with limited qubit connectivity is considered. Finally, it should be noted that throughout this paper it has been assumed that quantum circuits will be designed without reference to a target architecture and are thereafter not modified, apart from by adding swaps. Therefore, for example, the possibility that a quantum circuit which is to be executed is equivalent to an alternative circuit (in terms of the overall unitary operation performed), which is simpler once the swaps have been inserted is not considered. Not only does this assumption have the advantage that efficient circuit execution can be tackled as an abstracted problem which is largely independent of the actual quantum information processing (which is consistent with the typical approaches in the existing literature, and allows for elegant analytical results), but also corresponds to the emerging norm in quantum circuit compilation, where \textit{qubit routing} is addressed separately from quantum algorithm (and therefore circuit) design (e.g., \cite{me2, che} and \cite[Fig.~12]{sort}).
%
%
\section{Definitions}
\label{def}
To formalise the notion of the impact on the performance of a quantum algorithm by the insertion of swaps, raised in Section~\ref{intro}, it is necessary first to define the `dimensions' of a quantum circuit.
\begin{defn} The dimensions of a quantum circuit are its `width', the number of component qubits, denoted $n$; and its `depth', the number of layers of disjoint 2-qubit interactions, denoted $m$.
\end{defn}
As noted in Section~\ref{intro}, initially only full-width circuits are considered, meaning that the number of vertices in the interaction graph is also $n$. This enables the depth overhead to be formally defined:
\begin{defn}
\label{defref0}
The depth overhead, denoted $D$, is the multiplicative factor increase in the number of time steps (layers) required to run a specified quantum circuit on a quantum computer with limited qubit connectivity compared to that of a theoretical completely connected quantum computer. 
\end{defn}
In general, the depth overhead is a function of three terms: the graph of the qubit connectivity in the quantum computer, the quantum quantum circuit that is to be executed on the quantum computer, and the way that the quantum circuit is compiled to run on the quantum computer with limited qubit connectivity (i.e., by the insertion of swaps). In general the compiler can determine the initial qubit configuration (that is, it can perform an initial `placement' of qubits, so that the qubits in the quantum circuit to be executed are assigned to the physical qubits in the quantum computer in a manner which is `efficient' in some sense).\\
\indent The depth overhead is closely realted to the completely connected graph emulation complexity, as defined by Brierley \cite{Brierley}:
\begin{defn}
\label{defref1}
To emulate the completely connected graph is to send the qubit at vertex $a$ to $\pi(a)$ for all $a = 1,...n$ and any permutation $\pi : [1,n] \to [1,n]$.
\end{defn}
\begin{defn}
\label{defref2}
The completely connected graph emulation complexity, denoted $T$, is the number of time steps required to emulate the completely connected graph using a sequence of local gates.
\end{defn}
As noted in Section~\ref{intro}, an achievable completely connected graph emulation complexity is typically assumed to imply an achievable depth overhead, however a further condition is required:
\begin{condition}
\label{condref}
An achievable completely connected graph emulation complexity is sufficient to imply an achievable depth overhead if the graph can be decomposed into $\left\lfloor n/2 \right\rfloor$ disjoint connected pairs of vertices.
\end{condition}
This condition explicitly requires that the qubits can be permuted such that all interactions in a layer can be undertaken simultaneously, even for the extreme case where a layer consists of $\left\lfloor n/2 \right\rfloor$ 2-qubit interactions. 
\section{Main results}
\label{depth}
As noted below Definition~\ref{defref0}, the depth overhead, $D$, is a function of three terms: the graph of the qubit connectivity in the quantum computer, the quantum circuit that is to be executed on the quantum computer, and the way that the quantum circuit is compiled to run on the quantum computer with limited qubit connectivity (i.e., by the insertion of swaps). The first two of these can readily be expressed: let $\mathcal{G}_{n,r}$ be the set of all $r$-regular graphs with $n$ nodes; and let $\zeta$ be a quantum circuit. The third of these can also easily be expressed -- let $\chi$ be a compilation algorithm to compile any $\zeta$ to run on $G \in \mathcal{G}_{n, r}$ -- however, it requires a little further explanation: $\chi$ is an algorithm that for any quantum circuit, $\zeta$, outputs a second quantum circuit, $\zeta'$, that performs the same quantum operation, but can be performed on a given quantum computer with limited qubit connectivity. That is, $\chi$ is an algorithm (hereafter referred to as the \textit{compilation algorithm}) that inserts swap gates into $\zeta$. No memory or computational complexity restrictions are placed on the compilation algorithm, $\chi$, itself -- so it could simply be a look-up table which explicitly maps $\zeta$ to $\zeta'$ for all possible quantum circuits up to a certain maximum depth. In light of these further definitions, depth overhead is hereafter denoted by the function $D(G,\zeta,\chi) $.
\begin{thm}
$\forall r \geq 4 \,\, \exists(G \in \mathcal{G}_{n, r}, \chi) \,\,  \forall \zeta ,  \, D(G,\zeta,\chi) \leq 144 \log_2 n + \tilde{f}_{\mathcal{O}(1)}(n) \in \mathcal{O}(\log n)$, where $|\tilde{f}_{\mathcal{O}(1)}(n)| \in {O}(1)$ (which is also used throughout the remainder of this paper).
\end{thm}
Here the parenthesised $(G \in \mathcal{G}_{n, r}, \chi)$ is used to denote the fact that there exists a graph $G$, with the required connectivity, and an \textit{associated} compilation algorithm, $\chi$, to compile any quantum circuit to run on $G$. Therefore, the theorem states that there exists a qubit connectivity graph and associated compilation algorithm such that the depth overhead grows at most logarithmically with the number of qubits (the width of the circuit) when each qubit is connected to at least four others (and there are no other constraints on connectivity, such as a requirement of planarity).
\begin{proof}
Brierley \cite[Theorem~1]{Brierley}, gives a compilation algorithm for the \textit{Cyclic Butterfly} graph, which achieves $T(G) \in \mathcal{O}(\log n)$ with a width overhead of 2 (i.e., each application qubit is equipped with an ancilla). Thus to prove Theorem~3.1, Brierley's procedure is adjusted such that no ancillas are required (thus satisfying the full-width requirement), and such that Condition~\ref{condref} is satisfied, in order that $T(G) \in \mathcal{O}(\log n) \implies D(G) \in \mathcal{O}(\log n)$. Full details of the adjusted procedure are given in Appendix~\ref{app1}.
\end{proof}
Whilst Theorem~3.1 provides an important achievable depth overhead, the compilation algorithm detailed is rather contrived and the constant factor of 144 is rather high, therefore this is largelly for illustrative purposes, and in practise it is expected that there will be better ways to execute quantum circuits (i.e., still with a logarithmic depth overhead). Indeed, random regular graph theory suggests that it may well be the case that \textit{almost all} regular graphs (with degree greater than two) can achieve a logarithmic depth overhead with a much smaller constant factor. To formally state this possibility as a conjecture, it is first necessary to consider that the component qubits of any graph, $G \in \mathcal{G}_{n,r}$ are arranged in one of $n!$ permutations (from here on denoted \textit{network states}), thus for each, $G$ there exists a regular graph, denoted $G'$, with $n!$ vertices, each corresponding to one network state in $G$, and with connectivity corresponding to which other network states are reachable in a single time step (i.e., with a set of simultaneous swaps). Note that $G'$ \textit{is} regular because the connectivity of network states is a function of the permutations of edges of the underlying graph $G$ and not the current locations of the qubits (i.e., the network state -- a vertex of $G'$). Let $r'$ be the degree of $G'$, using the previous notation, $\mathcal{G}_{n!,r'}$ denotes the set of \textit{all} $r'$-regular graphs with $n!$ vertices, and let the subset containing all $G'$ be denoted $\mathcal{G}_{n!,r'}'$. Fig.~\ref{fnew} gives a very simple example of how $G'$ corresponds to $G$.
\begin{figure}[!t]
	\centering
	\includegraphics[width=0.75\textwidth]{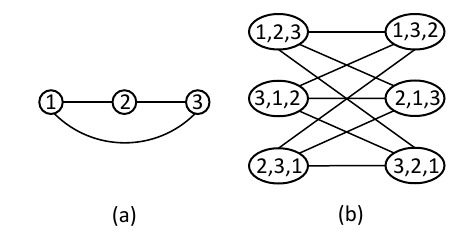}
	\captionsetup{width=.9\linewidth}
	\caption{\small{An illustration of how to construct a graph $G'$ corresponding to a $r$-regular graph $G$: (a) shows $G_{3,2}$, that is a 2-regular graph with three vertices -- the vertices are ordered left to right, and each vertex hosts an indexed qubit, so in the case shown the qubits are ordered 1,2,3; (b) shows the corresponding graph $G'$ which has $3! = 6$ vertices, and whose vertices correspond to qubit orderings, with edges showing that the two ordering are connected by a single layer of disjoint edge swaps (in this case, as there are three edges in $G$, only one edge can be swapped within a layer).}}
	\label{fnew}
\end{figure}
\begin{conjecture}
$\forall r \geq 3  \,\, \textit{almost all} \,\, G \in \mathcal{G}_{n, r} \, \, \exists \chi \,\, \forall \zeta, \, D(G,\zeta,\chi) \leq 4 \log_2 n + \tilde{f}_{\mathcal{O}(1)} (n)\in \mathcal{O}(\log n)$
\end{conjecture}
That is, rather than explicitly constructing a graph (such as the cyclic butterfly graph) and associated compilation algorithm to achieve logarithmic depth overhead, it is conjectured that such a depth overhead can be achieved on \textit{almost all} $r$-regular graphs\footnote{Note that \textit{almost all} $r$-regular graphs having some property is also referred to as a random $r$-regular graph \textit{almost surely} having that property.}. The validity of the conjecture is conditional on an upper-bound on the diameter of \textit{almost all} $r'$-regular graphs with $n!$ vertices (i.e., \textit{almost all} members of the set $\mathcal{G}_{n!,r'}$) also applying to \textit{almost all} members of the subset $\mathcal{G}_{n!,r'}'$. That is, the following argument assumes that there is no correlation between membership of the measure-0 subset for which the \textit{almost all} property does not apply, and membership of $\mathcal{G}_{n!,r'}'$ (itself a measure-0 subset of $\mathcal{G}_{n!,r'}$)\footnote{The requirement of \textit{no correlation} is rather strong, and may not hold in reality, however the conjecture will still hold with myriad weaker conditions, the strong condition merely stated for simplicity of exposition.}. This is the only unproven assumption in the following argument, and occurs in (\ref{eqa}).
\begin{proof}[Argument]
The conjecture relies on the equivalence $\forall G \in \mathcal{G}_{n,r}, \,  T(G) \equiv d_{G'}$, where $d$ is the diameter of the graph. That is, because $G'$ is the graph of network states, where network states that are reachable from each other in a single time step are connected, it can be seen that emulating the completely connected graph, as defined in Definition~2.4, exactly corresponds to traversing a path between some pair of vertices in $G'$, and thus the number of time steps to do this is upper-bounded by the diameter of $G'$.\\
\indent The diameter of a random regular graph, $G' \in \mathcal{G}_{n!,r'}'$, is, \textit{almost surely}, the least integer $d_{G'}$ satisfying $(r'-1)^{(d_{G'}-1)} \geq (2+ \epsilon) r' n! \log_e n!$, where $\epsilon > 0$ \cite[Section~2.4]{Bollobas}, i.e., for large graphs:
\begin{align}
(r'-1)^{(d_{G'}-1)} &\geq (2+ \epsilon) r' n! \log_e n! \nonumber \\
\label{eqa}
\implies d_{G'}-1 & \geq \frac{\log_2 ((2+ \epsilon) r' n! \log_e n!)}{\log_2(r'-1)} 
\end{align}
Given that $d_{G'}$ is the \textit{least integer} to satisfy (\ref{eqa}):
\begin{align}
d_{G'}-1 & \leq \frac{\log_2 ((2+ \epsilon) r' n! \log_e n!)}{\log_2(r'-1)} + 1 \nonumber \\
\implies d_{G'} & \leq \frac{\log_2 (2+ \epsilon) + \log_2 r' + \log_2 n! + \log_2 \log_e n!}{\log_2(r'-1)}+ 2 \nonumber \\
& = \frac{\log_2 r' + n \log_2 n - n \log_2 e + f_{\mathcal{O}(\log n)}(n)}{\log_2(r'-1)} + 2 \nonumber \\
& \leq \frac{n \log_2 (r+1) + n \log_2 n - n \log_2 e + f_{\mathcal{O}(\log n)}(n)}{\log_2(2^{n/4}-1)} + 2 \nonumber \\
& = \frac{n \log_2 (r+1) + n \log_2 n - n \log_2 e + f_{\mathcal{O}(\log n)}(n)}{\frac{n}{4} \log_2 2 + \log2 \left( \frac{2^{n/4}-1}{2^{n/4}} \right)} + 2\nonumber \\
& = 4 \log_2 n + \tilde{f}_{\mathcal{O}(1)}(n) \nonumber \\
\label{eq20}
& \in \mathcal{O}(\log n),
\end{align}
where $f_{\mathcal{O}(\log n)}(n) \in \mathcal{O}(\log n)$ and $r'$ is upper- and lower-bounded in the numerator and denominator using Lemmas~B.1 and B.2 respectively, and Stirling's approximation \cite{Dutka1991}, $\log_2 n! = n \log_2 n - n \log_2 e + f_{\mathcal{O}(\log n)}$, is also used.\\
\indent To complete the argument, notice that \textit{almost all} $r$-regular graphs are Hamiltonian connected \cite{hamreg}, therefore it is possible to separate the graph into disjoint connected pairs by the simple procedure of `pairing off' around a Hamiltonian cycle. Thus, Condition~\ref{condref} is satisfied, so $ T(G) \in \mathcal{O}(\log n) \implies D(G) \in \mathcal{O}(\log n)$, completing the argument.
\end{proof}
Emulation complexity results, concerning the reachability of certain network states from other distant network states, are useful for demonstrating achievable upper-bounds on the depth overhead, however this approach is insufficient for proving a lower-bound, as it neglects the possibility (at least in principle) of performing qubit interactions at any time (rather than always sorting such that an entire layer can be completed simultaneously). Nevertheless, it may appear that a lower-bound of $D(G) \in \Omega (\log n)$ is obviously true, given that the diameter of the graph must grow at least as $\log n$. However this is not, in fact, necessarily the case -- because prior qubit swapping could have been such that qubits are always `close' by the time they need to interact. Therefore a more general method is required to lower-bound the depth overhead. One such method is to treat the quantum computer as a disordered system, that must be ordered (to an extent) by performing swaps to achieve a gate. That is, it is necessary to perform swaps such that qubits that need to interact are adjacent in the interaction graph; however, once this interaction has occurred, then the two qubits concerned must next interact with other qubits, to which they will not, in general, be adjacent. So it is, that the interactions (quantum gates) have the effect of `destroying' the order (adjacency of qubits that need to interact) and so can be thought of as having a mixing effect on the state of quantum computer. This notion of the quantum computer as a disordered system can be formalised by an appropriate indexing of the interaction graph vertices, and also the qubits they host, such that the order required for an interaction to take place is clearly defined. This approach, it turns out, \textit{is} sufficiently general to lower-bound the depth overhead in $\Omega(\log n)$.
\pagebreak
\begin{thm}
$\forall (G \in \mathcal{G}_{n, r} , \chi) \,\, \exists \zeta \, , D(G,\zeta,\chi) \geq \frac{( 1 - (4/m)) \log_2 n}{2 \log_2 (r+1)} + \tilde{f}_{\mathcal{O}(1)}(n) \in \Omega(\log n)$.
\end{thm}
That is, for any proposed set of interaction graphs and associated compilation algorithm that claims to achieve a depth overhead that grows less than logarithmically with the quantum circuit width, it is possible to find a quantum circuit that invalidates the claim. Note that the circuit to do this must have at least five layers (so $1-(4/m) > 0$), however as the bound is expressed with an existential quantifier on the choice of circuit, it is already permitted to choose such a circuit and so the asymptotic bound holds without further quantification.
\begin{proof}
%
%
%
It is first necessary to formalise the idea of the quantum computer as a system that is mixed by interactions, and ordered by swaps. Informally, order is required for an interaction to occur because any two-qubits that are to interact (that is, undergo a two-qubit gate in the quantum circuit that is being executed) must be adjacent -- this order is imposed by a sequence of swaps. Directly following that interaction, the two qubits concerned may next have to interact with any of the other qubits, so further swapping (ordering) is required to enable subsequent interactions, and so it follows that an interaction may be thought of as introducing disorder.\\
\indent This idea can be formalised by indexing and labelling the vertices and qubits: the state of the system is described by indexing each vertex in the interaction graph, and also indexing each qubit so that each vertex has exactly one local qubit. Performing a swap on two qubits has the effect of updating the correspondence between the vertex indices and qubit indices accordingly. For example, if qubits $Q_1$ and $Q_2$ are located at vertices $V_1$ and $V_2$ respectively, and a swap is performed between $V_1$ and $V_2$ then the result will be that $Q_1$ is located at $V_2$ and $Q_2$ at $V_1$.\\
\indent Additionally, each qubit is labelled with the index of the next qubit with which it will interact (its \textit{target} qubit), and a \textit{flag} to indicate whether this interaction is available at the present time, i.e., the flag is a Boolean which is false (\texttt{F}) if the target qubit must first do at least one other interaction, or otherwise true (\texttt{T}). For example, if a gate between $Q_1$ and $Q_2$ is scheduled followed by a gate between $Q_1$ and $Q_3$, and $Q_3$ has no other gates beforehand, then $Q_1$ will be labelled with its target, $Q_2$, and $Q_2$ labelled with its target, $Q_1$, and the flag of each will be \texttt{T}, as they have an interaction that is available. In the case of qubit $Q_3$, its target is $Q_1$, so it will be labelled $Q_1$, however its flag will be \texttt{F}, as there is another interaction (i.e., between $Q_1$ and $Q_2$) which must occur first.\\
%
%
%
%
%
%
%
\indent This set-up can be used to express the entropy added to the system for each interaction (gate)\footnote{It is important to appreciate that `entropy' in this context is solely concerned with the process of ordering qubits by swapping being described (i.e., a \textit{classical} quantity) and has nothing to do with the actual quantum interactions possibly changing the \textit{quantum} entropy of the actual quantum state is being computed.}. Let $Q_1$, located at vertex $V_1$, be a qubit with an interaction in the next layer (i.e, its flag is \texttt{T}). All that is required for this interaction to take place is that the system must be ordered such that the qubit labelled with target $Q_1$ and flag \texttt{T} is located at a vertex, $V_2$, which is adjacent to $V_1$. This is summarised in Table~\ref{tab1} (in which the blank elements signify that the value is arbitrary, for example it is not important what the index of the qubit at $V_2$ is, just that it is the one which is targetting $Q_1$ for the next interaction).\\
\indent The increase in entropy caused by an interaction is concerned with the change in system state following that interaction. From Table~\ref{tab1} and its caption, it can be seen that there is only one entropy-increasing change: the relabelling of the qubit at vertex $V_2$, which corresponds to some random process (this will be formalised shortly). However, it is also relevant that, given that $V_1$ has $r$ neighbours (i.e., it is a $r$-regular graph) there are $r$ possibilities for $V_2$. Thus, letting this be the $i${th} interaction and $\Delta S_{V_2}^{(i)}$ be the increase in entropy associated with the relabelling of $V_2$ from $Q_1$ before the interaction (as was required in order for the interaction to be possible) to some random label $Q_x$ after the interaction, a lower-bound on the entropy added to the system by the $i$th interaction, $\Delta S_I^{(i)}$, can be expressed:
%
\begin{equation}
\label{eq401}
\Delta S_I^{(i)} \geq \Delta S_{V_2}^{(i)} - \log_2 r,
\end{equation}
where the term `$-\log_2 r$' is included as there are $r$ vertices adjacent to $V_1$ and thus $r$ possibilities for $V_2$ (as previously noted), and this is therefore an inequality as a uniform distribution over the possibilities for $V_2$ is assumed, but may not necessarily be the case. The \textit{same} ($i$th) interaction can also be viewed from the perspective of a qubit, $Q_2$ at $V_2$, as summarised in Table~\ref{tab2}. Following the same rationale:
\begin{equation}
\label{eq402}
\Delta S_I^{(i)} \geq \Delta S_{V_1}^{(i)} - \log_2 r.
\end{equation}
Thus putting (\ref{eq401}) and (\ref{eq402}) together yields:
\begin{equation}
\label{eq403}
\Delta S_I^{(i)} \geq \frac{1}{2} \left( \Delta S_{V_1}^{(i)} + \Delta S_{V_2}^{(i)} \right) - \log_2 r.
\end{equation}
In order to consider the re-labelling of a qubit to be a random event, it is necessary to consider not a single `known' quantum circuit, but rather a random quantum circuit sampled from some defined ensemble. The ensemble in question is chosen to be the ensemble of random circuits of some depth $m$, in which each layer consists of $\lfloor n/2 \rfloor$ interactions -- that is, the maximum possible number of interactions occurs in each layer (and the $m$ layers of the circuit are chosen independently and identically distributed (i.i.d.)). In the case where $n$ is odd, the qubit which is not involved can be thought of as pairing up with a further virtual qubit, therefore there are $(2 \lceil n/2 \rceil)!! = ((2 \lceil n/2 \rceil)!)/(2^{\lceil n/2 \rceil}(\lceil n/2 \rceil)!)$ permutations of interactions in each layer. A random layer is sampled uniformly from this ensemble, which means that the random process has an entropy of:
\begin{table*}[!t]
	\centering
	\begin{tabular}{ l l l l l l} 
		 & \multicolumn{2}{c}{\textbf{Before interaction}} & & \multicolumn{2}{c}{\textbf{After interaction}} \\
		\hline\hline 
		\textbf{Vertex} & $V_1$ & $V_2$ & & $V_1$ & $V_2$ \\
		\textbf{Qubit} & $Q_1$ &  & & $Q_1$ &  \\
		\textbf{Label} &  & $Q_1$ & &  & $Q_x$ \\
		\textbf{Flag} & \texttt{T} & \texttt{T} & & $\mathrm{f}_1$ & $\mathrm{f}_2$ \\
	\end{tabular}
	\captionsetup{width=.9\linewidth}
	\caption{Relabelling of qubits associated with an interaction, $\mathrm{f}_1 , \mathrm{f}_2 \in \{\text{\texttt{F}},\text{\texttt{T}}\}$. Note that $\mathrm{f}_1 , \mathrm{f}_2$ are determined by the system state once the interaction is done, and thus add no entropy. }
	\vspace{0.3cm}
	\label{tab1}
		\begin{tabular}{ l l l l l l} 
			& \multicolumn{2}{c}{\textbf{Before interaction}} & & \multicolumn{2}{c}{\textbf{After interaction}} \\
			\hline\hline 
			\textbf{Vertex} & $V_1$ & $V_2$ & & $V_1$ & $V_2$ \\
			\textbf{Qubit} &  & $Q_2$ & &  & $Q_2$ \\
			\textbf{Label} & $Q_2$ &  &  & $Q_y$ &  \\
			\textbf{Flag} & \texttt{T} & \texttt{T} & & $\mathrm{f}_1$ & $\mathrm{f}_2$ \\
		\end{tabular}
		\captionsetup{width=.9\linewidth}
		\caption{Relabelling of qubits associated with an interaction from the perspective of the qubit at $V_2$. }
		\label{tab2}
\end{table*}
\begin{equation}
\label{eq400}
\Delta S_L = \log_2 (2 \lceil n/2 \rceil)! - \left\lceil \frac{n}{2}\right\rceil - \log_2 (\lceil n/2 \rceil)!,
\end{equation}
where all entropies are expressed in bits\footnote{Note that the nomenclature `$\Delta S_L$' is used, rather than simply `$S_L$', as this will later correspond to the \textit{change} in entropy caused by relabelling the qubits after an entire layer of interactions}.\\
\indent A layer of a quantum circuit is simply a list of qubits which must interact, and so there is an exact deterministic correspondence between sampling from this random process (which returns a random layer) and the process of relabelling qubits following an entire layer of interactions. It follows that, if a $m$ layer quantum circuit, generated by i.i.d. sampling $m$ random layers, is to be executed, then from (\ref{eq403}), the total increase in entropy associated with the process of qubit relabelling can be expressed:
%
%
\begin{align}
\sum_{i=1}^{N_I} \Delta S_I^{(i)} & \, \geq \, \frac{1}{2}\sum_{i=1}^{N_I} \Delta S_{V_1}^{(i)} + \Delta S_{V_2}^{(i)} - N_I \log_2 r \nonumber \\
\label{eq420}
& \,=\, \frac{m}{2} \Delta S_L - N_I \log_2 r,
\end{align}
where $N_I = m \lfloor n/2 \rfloor $ is the total number of interactions. This entropy is not all added at once, but rather is added when the interactions occur, and thus the interacting qubits that had been ordered (swapped) to be adjacent are relabelled with their new targets.\\ 
%
%
\indent The purpose of expressing the entropy increase of updating of the qubit labels when qubits undergo interactions is to lower-bound the number of ordering time steps (layers of swaps) needed to offset this disordering. To do this, it is necessary to consider the entropy of the initial and final states of the system. Starting with the initial state, according to the text below Definition~2.2, the compiler is free to set the initial qubit locations, so the contribution to the initial entropy by the relation between the vertex indices and qubit indices, denoted $S_{init_q}$, can only be lower-bounded as greater or equal to zero. Added to this, the qubits are labelled for the first layer of interactions, according to the process detailed above, leading to (\ref{eq400}). This allows the initial entropy of the system, $S_{init}$, to be expressed:
\begin{align}
S_{init} & \, = \, S_{init_q} + \Delta S_L \nonumber \\
\label{eq445}
 & \, \geq \, \log_2 (2 \lceil n/2 \rceil)! - \left\lceil \frac{n}{2}\right\rceil - \log_2 (\lceil n/2 \rceil)!
\end{align}
Turning to the final state, it is helpful to consider a hypothetical $(m+1)$th layer that is not actually executed, but is there so that the qubits are relabelled with new target qubits after their $m$th layer interactions. The interactions in the $(m+1)$th layer are randomly generated according to the same process as for the first $m$ layers, and so the entropy of the labels in the final state is equal to $\Delta S_L$. As these additional interactions are not actually executed, this does not increase the ordering needed, and the resultant labelling in the final state of the circuit (equal to $\Delta S_L$) is accounted for in the following analysis. Regarding the correspondence between the qubit and vertex indices, unlike for the initial case, there need not be any correlation between the qubit and vertex indices, which can be expressed by $S_{fin_q} \leq \log_2 n!$, i.e., a uniform distribution over all permutations of qubit locations. Thus the final entropy, $S_{fin}$ can be upper-bounded: 
\begin{align}
S_{fin} & \, = \, S_{fin_q} + \Delta S_L \nonumber \\
\label{eq430}
 & \,\leq \, \log_2 n! + \log_2 (2 \lceil n/2 \rceil)! - \left\lceil \frac{n}{2}\right\rceil - \log_2 (\lceil n/2 \rceil)!,
\end{align}
where (\ref{eq430}) is an inequality because it considers the upper-bounding case where the final arrangement of the qubits is maximally mixed (i.e., as noted, $S_{fin_q} \leq \log_2 n!$, with equality when the final arrangement of qubits is distributed uniformly over all possibilities). Bounding the initial and final entropy, in (\ref{eq445}) and (\ref{eq430}) respectively, and the entropy of the qubit relabelling process in (\ref{eq420}) allows an overall entropy balance to be expressed, which much hold\footnote{As the system has a finite number of degrees of freedom, there is a corresponding finite maximum disorder, and thus the ordering process must also be such that this is never exceeded, however this property is not required for the proof.}:
\begin{equation}
\label{eq440}
S_{fin} = S_{init} + \sum_{i=1}^{N_I} \Delta S_I^{(i)} + \Delta S(\chi),
\end{equation}
where $-\Delta S(\chi)$ is the entropy reduction achieved by the compilation algorithm, which is hereafter referred to as \textit{ordering process} to emphasise its role in this entropic formulation of circuit compilation. This ordering process is the sequence of swaps performed such that qubits that are to interact are adjacent, for example, in the cyclic butterfly interaction graph addressed by Brierley, the ordering process is the method of sorting the rows and columns such that a layer of interactions could be executed. General compilation algorithms have the same essential property: ordering the system by swapping qubits such that qubits that are to interact are adjacent in the interaction graph. Rearranging (\ref{eq440}), and then substituting in (\ref{eq400}), (\ref{eq420}), (\ref{eq445}) and (\ref{eq430}):
\begin{align}
-\Delta S(\chi) & \, = \,   \sum_{i=1}^{N_I} \Delta S_I^{(i)} + S_{init} - S_{fin} \nonumber \\
&\, \geq \, \frac{m}{2} \left( \log_2 (2 \lceil n/2 \rceil)! - \left\lceil \frac{n}{2}\right\rceil - \log_2 (\lceil n/2 \rceil)! \right) - \frac{mn}{2} \log_2 r - \log_2 n!  \nonumber \\
\label{eq450}
& \, \geq \,  \left( \frac{m}{4}-1 \right)  n \log_2 n  - \frac{mn}{2} \log_2 r + \tilde{f}_{\mathcal{O}(n)}(n)
\end{align}
using Stirling's approximation \cite{Dutka1991} and where $|\tilde{f}_{\mathcal{O}(n)}(n)| \in \mathcal{O}(n)$ by definition. Note that the term $\frac{mn}{2} \log_2 r$ has not been absorbed into $\tilde{f}_{\mathcal{O}(n)}(n)$ as in Section~\ref{general} a situations where $r$ grows with $n$ is considered.\\
\indent It follows that it is necessary to consider the extent to which the ordering process, $\chi$, can be compressed into time steps. That is, to lower-bound the number of time steps required to achieve the necessary reduction in entropy $\Delta S(\chi)$. To do so, it is first important to appreciate that some circuits can be executed without any ordering (swaps), i.e., those which, by chance, have interacting qubits already adjacent in the interaction graph (for every layer in the circuit).\\
\indent To put it another way, a swap is not, in and of itself, an order inducing event, but rather it is the \textit{decision to swap or not}, that imparts order. After each interaction takes place, there is one scenario in which the ordering process actually requires no further time steps, namely if the decision is to make no swaps at all. This can be accounted for by considering a decision to take place after each interaction, as to whether any swapping is required. Notably, these decisions are ordering events that require no time steps. As these are binary decisions (to do some swaps or not), they each correspond to a maximum of one bit of order, and occur after each interaction, i.e., $N_I$ times in total. Letting the entropy \textit{removed} by these decisions about whether to swap be denoted $- \Delta S_{SD}$, this can be bounded:
\begin{equation}
\label{eq44a}
 -\Delta S_{SD} \leq N_I = m \left\lfloor \frac{n}{2} \right\rfloor \leq \frac{mn}{2},
\end{equation}
Now consider the maximum entropy that can be removed from the system in a time step, denoted $-\Delta S_{S}$ (i.e., when it has been decided that swaps \textit{are} needed). Lemma~B.1 upper-bounds the number of qubit arrangements that can be reached in a single time step (i.e., a single set of disjoint edge swaps) as $(r+1)^n$, and thus the entropy reduction of a single time step of ordering is upper-bounded by taking a uniform distribution over these qubit arrangements, that is:
\begin{equation}
\label{eq44}
 -\Delta S_{S} \leq n \log_2 (r+1). 
\end{equation} 
It follows that the number of time steps, $N_S$, required for the ordering process, $\chi$, (for the ensemble) can be lower-bounded:
\begin{align}
-\Delta S(\chi) & \, = \, -N_S \Delta S_{S} - \Delta S_{SD} \nonumber \\
& \, \leq \, N_S n \log_2 (r+1) + \frac{mn}{2} \nonumber \\
\label{eq45}
\implies N_S & \geq \, \frac{-\Delta S(\chi) - (mn/2)}{n \log_2 (r+1)}.
\end{align}
It is important to appreciate that, as this upper-bound on the entropy removed by the ordering process is derived by simply considering the maximum order that can be added in a layer of swaps, it holds even if a sophisticated compilation technique is used, that takes into account the entire quantum circuit to be executed, and not merely the simple description of the current state of the system (i.e., as defined above (\ref{eq401})). To put it another way, when expressing the disorder brought about by executing the ensemble of quantum circuits (i.e., (\ref{eq420})) any order encoded in the system corresponding to other interactions that will need to be executed in future layers has effectively been marginalised out.\\
\indent Substituting (\ref{eq450}) into (\ref{eq45}), and considering circuits which have at least five layers in the final step (i.e., so $m \geq 5$):
\begin{align}
N_S  & \, \geq \, \frac{\left( \frac{m}{4}-1 \right)  n \log_2 n  - \frac{mn}{2} \log_2 r + \tilde{f}_{\mathcal{O}(n)}(n)}{n \log_2 (r+1)}  \nonumber \\
& \, = \, \frac{ \frac{m-4}{2}  \log_2 n  - m \log_2 r }{2 \log_2 (r+1)} + \tilde{f}_{\mathcal{O}(1)}(n) \nonumber \\
& \, = \, \frac{( 1 - (4/m)) \log_2 n}{2 \log_2 (r+1)} + \tilde{f}_{\mathcal{O}(1)}(n) \nonumber \\
\label{eq46}
& \, \in \, \Omega (\log n).
\end{align}
where $|\tilde{f}_{\mathcal{O}(1)}(n)| \in \mathcal{O}(1)$, as defined in Theorem~3.1. As $N_S$ time steps are required to order the ensemble, at least one of the quantum circuits in the ensemble must take at least $N_S$ time steps to be completed. By definition, the depth of each quantum circuit in the ensemble is $m$ (a constant), therefore there exists a quantum circuit with depth overhead $\in \Omega (\log n)$, proving Theorem~3.3.
\end{proof}
The entropic approach to proving Theorem~3.3 is useful for formalising the idea of swapping being an ordering process, and interactions being a mixing process -- and that the goal of quantum circuit compilation is to continually introduce order so the desired quantum operations (interactions) can be executed. In particular, treating the quantum computer as a disordered system seems to be especially apt for addressing the question of what can be achieved in terms of depth reduction by initial qubit placement. In this paper this is manifested in the term $( 1 - (4/m))$ in the numerator of (\ref{eq46}), which implies that the possibility of quantum circuits with four layers or fewer being executed with constant depth overhead (that is, by virtue of initial qubit placement) has not been precluded. This figure of four-layers or fewer suffices for the purposes of this paper, but further research may refine this bound further and the entropic approach described herein may help to do so.\\
\indent However, whilst the entropic approach has its advantages, it is worth noting that because the proof uses uniform distributions over the possibilities, the core of the proof is essentially equivalent to a simpler counting argument: let $n_u(m)$ be the number of \textit{uncompiled} quantum circuits of depth $m$, which is at least $(n!!)^m$; let $n_c(N_S)$ be the number of \textit{compiled} quantum circuits of depth $N_S$ (where compiled means that swaps have been inserted to enable the circuit to be executed on the target architecture). In any one time step, a given qubit can either idle, or swap with one of $r$ neighbours, or interact with one of $r$ neighbours, therefore there are at most $((2r+1)^n)^{N_S}$ compiled quantum circuits of depth $N_S$ (this is essentially a restatement of Lemma~B.1, explicitly including the possibility of interacting, which was not needed for the main proof). Noting that each compiled quantum circuit executes exactly one uncompiled quantum  (so $n_c(N_S) \geq n_u(m)$), these two inequalites: $ n_u(m) \geq (n!!)^m $ and $n_c(N_S) \leq (2r+1)^{nN_S}$ can be combined to give $N_S / m \in \Omega (\log n / \log r )$, for at least some quantum circuit.  
%
%
\section{Generalising the depth overhead lower-bound}
\label{general}
It is of additional interest to generalise the lower-bound on the depth overhead for some situations of practical interest, namely when the interaction graph is not regular and the average and maximum degree can grow with $n$; and when not all of the vertices are filled with application qubits, but instead some hold ancillas which can be used to assist the routing. Addressing first the case where the graph is not regular, let $r_{ave}(n)$ and $r_{max}(n)$ be respectively the average and maximum degrees of the vertices of the graph. Let $0 \leq \tilde{a} < (1/2)-(2/m)$, and let $a$ be a constant (note that the following bound holds for all $a$, but is only useful for $0 \leq a < 1$). To allow for irregular graphs, the notation is extended slightly such that $\mathcal{G}_{n, r_{ave}, r_{max}}$ represents all graphs with $n$ nodes, with an average degree of $r_{ave}$ and a maximum degree of $r_{max}$. In the following, it is helpful to adjust the bound in Section~3 slightly, specifically, the entropy reduced in a single step of swaps can be upper-bounded by considering the scenario in which, at each time step, each of the $nr/2$ edges is independently set into a specified state (i.e., intentionally swapped or not) thus achieving a reduction in entropy of $(nr/2)$ bits (each edge decision adds at most one bit of order, as there are two states). This enables (\ref{eq44}) to be replaced with:
\begin{equation}
\label{eq44new}
 -\Delta S_{S} \leq \frac{nr}{2},
\end{equation} 
which in turn means that the second line in (\ref{eq46}) will read:
\begin{equation}
\label{eq46new}
N_S  \geq   \frac{ \frac{m-4}{2}  \log_2 n  - m \log_2 r }{r} + \tilde{f}_{\mathcal{O}(1)}(n) 
\end{equation}
\begin{thm}
	$\forall r_{ave}(n) \in \mathcal{O}(\log^a n) \,\, \forall r_{max}(n) \in \mathcal{O}( n^{\tilde{a}}) \,\, \forall (G \in \mathcal{G}_{n, r_{ave}, r_{max}}, \chi) \,\, \exists \zeta \, , D(G,\zeta,\chi)  \in \Omega(\log^{1-a} n)$
\end{thm}
That is, even if the average and maximum degrees of the interaction graph are allowed to increase very mildly (sub-logarithmically) with the number of vertices, there remains a lower-bound on the depth overhead.
\begin{proof}
The bound in (\ref{eq420}) clearly still holds if the maximum degree is considered (i.e., $r$ is replaced with $r_{max}$); by contrast, (\ref{eq44new}) simply sums up the total number of edges, so here $r$ should be replaced with $r_{ave}$. This means that (\ref{eq46new}) can be re-written:
\begin{align}
N_S  & \, \geq \, \frac{\frac{m-4}{2}  \log_2 n  - m \log_2 r_{max}}{r_{ave}} + \tilde{f}_{\mathcal{O}(1)}(n) \nonumber \\
\label{gen10}
& \, \in \, \Omega(\log^{1-a} n)
\end{align}
for the restrictions on $r_{max}$ and $r_{ave}$ given in the theorem.
\end{proof}
Notably, as $m$ gets large, i.e. for deep circuits, $(1/2)-(2/m)$ becomes close to $1/2$, and so the maximum degree can grow almost with the square root of the number of vertices without reducing the lower-bound on the required depth overhead (that is, if the average degree does not increase with the number of qubits).\\
\indent Turning now to the situation where ancillas are allowed to aid the circuit execution, let $n_v \geq n$ be the number of vertices, of which only $n$ host application qubits. The method developed in Section~\ref{depth} is not sufficiently powerful to incorporate all the ways in which one may use these ancillas to speed up the execution of quantum circuits, but it is nevertheless of interest to characterise the extent to which ancillas can help when only swap gates are allowed to be added to the quantum circuit.
\begin{thm}
	$\forall n_v \geq n \,\, \forall r_{max}(n) \in \mathcal{O}(\log^a n) \,\,  \forall (G \in \mathcal{G}_{n, r_{ave}, r_{max}}, \chi) \,\, \exists \zeta \, , D(G,\zeta,\chi)  \in \Omega(\log^{1-a} n)$, where only swap gates may be added to the quantum circuit.
\end{thm}
That is, the introduction of ancilla qubits to aid the compilation algorithm only yields a constant factor improvement in performance, at best, if only swaps are added to achieve the compilation. To see this, let $a=0$ so that the graph is regular (or more precisely, has constant maximum degree), in which case the bound is equal to that of Theorem~3.3.
\begin{proof}
This can be seen by first noting that the additional vertices (without application qubits) have no affect on the entropy addition associated with relabelling. Furthermore, regardless of $n_v$, there are a maximum of $nr_{max}$ edges connected to at least one application qubit. As other vertices will be occupied by unlabelled ancillas, swapping the connecting edges thereof will not reduce the disorder of the system. Thus, (\ref{eq44new}) can be re-written:
\begin{equation}
\label{gen20}
-\Delta S( \chi) \leq n r_{max},
\end{equation}
this allows (\ref{eq46new}) to be re-written:
\begin{align}
N_S  & \, \geq \,  \frac{ \frac{m-4}{2}  \log_2 n  - m \log_2 r_{max} }{r_{max}} + \tilde{f}_{\mathcal{O}(1)}(n) \nonumber \\
\label{gen21}
& \, \in \, \Omega(\log^{1-a} n),
\end{align}
\end{proof}
\indent The restriction on $r_{max}$ is necessary as if, for example, the number of vertices is allowed to be arbitrarily large, and $r_{max}$ grows too fast with $n$, then it would be possible to have a richly connected region of the quantum computer in which all the application qubits always reside, whilst having additional sparsely connected regions which are unused and exist only to achieve $r_{ave}$. Other lower-bounds, expressed in terms of $r_{ave}(n)$ could be found should an appropriate restriction be placed on $n_v(n)$.
\section{Discussion}
\label{disc}
Theorem~3.1 proves that a depth overhead in $\mathcal{O}(\log n)$ is achievable, and Theorem~3.3 provides the converse, that quantum circuits can be constructed such that the depth overhead must be at least $((1-(4/m)) \log_2 n)/(\log_2 (r+1))$ (plus or minus lower order terms). Thus together Theorem~3.1 and 3.3 constitute a tight bound, in terms of the order of the asymptotic growth of the depth overhead. The existence of the converse confirms that the procedure and graph proposed by Brierley \cite{Brierley} (with the adjusted procedure detailed in Appendix~\ref{app1}) is asymptotically efficient. Intriguingly, however, Conjecture~3.2 suggests that such performance may be achievable on almost all $r$-regular graphs, where $r \geq 3$.\\
\indent Broadly speaking, these results correspond to what one might intuitively expect. The diameter of any $r$-regular graph grows at least with $\log n$ and thus, notwithstanding the formal justification for using the entropic approach, it is hard to see how a compilation algorithm would achieve a lower depth overhead. Conversely, though, Benes networks \cite{benes} have been well-known in the classical literature for some time and it is therefore unsurprising that this idea can be adapted to the idea of qubit routing to achieve a logarithmic depth overhead, as in \cite{Brierley} and Theorem~3.1. The generalised framework considered in Section~\ref{general} can also be used to show that the entropic method used does not impose a minimum depth overhead on the completely connected graph: that is, a $(n-1)$-regular graph with no parallel edges or loops. This can be seen by substituting $r=n-1$ into the penultimate line of (\ref{eq46}), which leaves the entire bound as $O(1)$. This is as expected, and provides an important check on the proof, as any quantum circuit can be executed on a completely connected quantum computer without additional swaps being added (so has a depth-overhead of one). Another notable property of the entropic method of proving the depth overhead lower-bound, is that it can be used to infer the amount that the initial ordering of the qubits can help with the qubit routing: Theorem~3.3 only proves the lower-bound for quantum circuits with at least 5 layers, so for circuits with four layers or fewer, then it may be possible to use the initial arrangement of qubits to achieve most of the interactions with few additional swaps.\\
\indent The results presented in this paper are also significant for informing the design of near-term quantum computers, as well as the process of ordering qubits to achieve gates therein. Notably, the qubit interaction graphs used to achieve depth overhead in $\mathcal{O}(\log n)$ in both Theorem~3.1 and Conjecture~3.2 cannot be embedded in $k$-dimensional space. The best known achievable depth overhead for a qubit interaction graph that is embeddable in $k$-dimensional space is in $\mathcal{O}(\sqrt[k]{n})$, and thus there may be good reason to favour quantum computer designs which do not constrain the qubit interaction graph in this manner, such at NQIT's proposed quantum computer, in which the graph edge are given physical reality as photonic links \cite{Nickerson}.\\
%
%
\indent Finally, Theorem~4.1 shows that a single highly connected node will not necessarily help reduce the depth overhead (as one may expect), but perhaps less obviously, Theorem~4.2 shows that allowing the total number of qubits to exceed the number of application qubits by an amount that grows arbitrarily fast (with $n$) does not improve the asymptotic performance, if the graph is still constrained to be regular, and only swap gates are allowed to be added to the circuit. However, by only allowing swaps, Theorem~4.2 restricts the manner in which the quantum circuit representing the algorithm can be supplemented to enable it to run on a quantum computer with limited connectivity, whilst still achieving the original computation. In a full-width quantum circuit, i.e. in which all qubits are application qubits, swapping is indeed the only mechanism available, however if there are ancilla qubits available for use in executing the quantum circuit then there is the possibility of using teleportation to achieve the desired interactions. Indeed, Rosenbaum shows that if the number of ancillas available grows as the number of application qubits squared, then a constant depth overhead can be achieved \cite{Rosenbaum}. Therefore an important future research direction is to establish to what extent, exactly, the presence of ancillas can reduce the depth overhead.
%
%
\section*{Acknowledgements}
This work was supported by an Networked Quantum Information Technologies Hub Industrial Partnership Project Grant, and thanks are also given to 
Sathyawageeswar Subramanian who proof-read an early version of the paper, and made several useful suggestions concerning the rigour and clarity of the paper. Additionally, the contributions of the anonymous reviewer are much appreciated, especially pointing out the equivalent counting argument for Theorem~3.3, which in turn prompted the reult to be tightened slightly by reusing Lemma~B.1 in the bound of the entropy reduction in a time step.
%
%
\appendix
\section{Achieving $D \in \mathcal{O}(\log n)$}
\label{app1}
Let there be a cyclic butterfly network with $n/8$ vertices.
\begin{enumerate}
	\item Use the first alteration in \cite[Section~IV]{Brierley} to make a graph with degree 3. That is, each vertex in the cyclic butterfly (degree 4) can be replaced with a ring of four vertices, each connected to one of the original four edges, yielding a degree three graph with $n/2$ vertices, but incurring a factor two increase in the time overhead. This incurs a factor two increase in the depth overhead, later taken into account.
	\item For each vertex, connect another vertex (i.e., yielding a graph with $n$ vertices, half of which are of degree 4, half of which are of degree 1). These newly added vertices, along with their originally connected existing vertex are referred to as `pairs' below.
	\pagebreak
	\item Add arbitrary extra edges to make the graph $r$-regular (where $r \geq 4$) as specified. These edges are unused.\footnote{It is always possible to make the graph \textit{exactly} 4-regular, by adding these extra edges such that the extra vertices added in Step~2 are connected in a ring, in exactly the same manner that the original vertices were in Step~1. In the case where the graph is to be $r$-regular, for some $r>4$, strictly speaking it is necessary to check that the total number of vertices is such that a $r$-regular graph can be constructed without disturbing the cyclic-butterfly structure that fixes some of the edges, but this minor detail is unlikely to be important in practise.}
	\item For all pairs of vertices, if the paired qubits are to interact in the current layer then do the interaction, and remove from the layer. \textbf{1 time step maximum}
	\item Choose a random vertex from the set of vertices in the original graph, label the qubit at this vertex `A'. Label the qubit at the paired vertex (i.e., as added in step 2) `B'. If `A' is due to interact in the current layer, label its interacting qubit `A', likewise if `B' is due to interact in the current layer, label its interacting qubit `B'. Label the paired qubit of the newly labelled `A' as `B', and label the paired qubit of the newly labelled `B' as `A'. Continue until the cycle terminates, either because both ends of the string have reached qubits which do not interact with any other qubit in the current layer, or because a loop is formed. Choose another random starting point and repeat -- continue until all qubits are labelled.
	\item The network is now labelled such that each pair has at most one A and one B. Furthermore, all As are to interact with other As if they have a interaction in the current layer (and likewise for Bs).
		\begin{figure}[!t]
		\centering
		\includegraphics[width=0.9\textwidth]{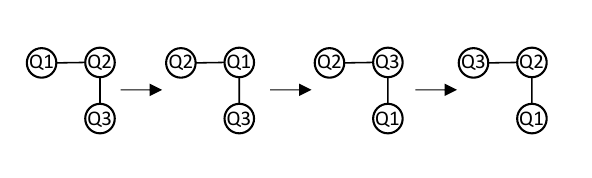}
		\captionsetup{width=.9\linewidth}
		\caption{\small{Demonstrating swap of Q1 and Q3 in three time steps, without disturbing Q2 -- i.e., Q3 acts as the ancilla, into which Q1 is loaded.}}
		\label{f1}
	\end{figure}
	\item Order the network so that all of the As are in vertices in the original graph, and Bs are in the additional paired vertices. \textbf{1 time step maximum}
	\item Let the second of each pair of application qubits (i.e., the one with only one connection to the first application qubit) act as the ancilla for the subsequent sorting. Note that in the compilation algorithm detailed by Brierley \cite{Brierley} `moving a qubit into a neighbours ancilla' is an identical operation to swapping an ancilla. Thus, this can be achieved by the process shown in Fig.~\ref{f1}. This takes three time steps, and at worst has to occur three times to achieve the required swap with the ancilla for all qubits (i.e., label qubits alternately as odds and evens in each string of ancilla swaps, do all of the odds and then all of the evens, and at worst one further step will be required in the case of qubits connected in odd length cycles). The compilation algorithm detailed by Brierley \cite{Brierley} takes $8 \log_2 n$ time steps (i.e., taking into account the factor two increase in depth overhead associated with reducing the graph degree previously noted). Thus into total, \textbf{a maximum of} $3 \times 3 \times 8 \log_2 n = 72 \log_2 n$ \textbf{time steps}
	\item The vertices of the original graph are arranged in loops of four, which means that the original graph can be divided into $n/4$ connected pairs of vertices: Condition~\ref{condref} is satisfied, and the sort in the previous step can be such that all of the interactions of qubits labelled A can be conducted.
	\item Notice that the ordering is such that a qubit labelled A is at each of the original nodes, this means that a qubit labelled B is at each of the extra vertices introduced in step 2. Thus a swap is conducted at each pair such that the Bs now occupy the original vertices, and steps 8 -- 9 are repeated with the qubits labelled B, thus completing the layer. \textbf{1 time step}
\end{enumerate}
The procedure requires a maximum of $3 + 2 \times 72 \log_2 n = 3 + 144 \log_2 n$ time steps, i.e., noticing that step 8 is conducted twice. Therefore the need for ancillas has been removed, and the number of time steps remains $\mathcal{O}( \log n)$.
\section{Lemmas}
\begin{lemma} $r' \leq (r+1)^n$
\end{lemma}
\begin{proof}
	This inequality can be seen by considering the upper-bounding case where each of the $n$ nodes is independently selected to either permute with one of $r$ neighbours, or to remain inactive for the time step, i.e., $r+1$ possibilities. Therefore each state can transfer to at most $(r+1)^n$ other states, which is the degree of $G' \in \mathcal{G}_{n,r'}'$ by definition. 
\end{proof}
\begin{lemma} $r' \geq 2^{n/4}$
\end{lemma}
\begin{proof}
Consider the set of edges of $G$, which has size $nr/2$, a subset of which can be constructed of size $n/4$ consisting of edges that can be independently permuted (or not) by the following method: Select an edge, remove all edges connected to either connecting vertex from future consideration, repeat until all edges are either selected or removed from future consideration. At most $2r-2$ edges are removed from future consideration (i.e., if non of the connecting edges have already been neglected), therefore each step occupies (either by selecting or neglecting) at most $2r-1 < 2r$ edges. Therefore a set of size greater than $(nr/2)/2r = n/4$ has been constructed, containing combinations of edges which can be independently permuted or left, thus there are at least $2^{n/4}$ possible state transitions available in each state of $G'$, and Lemma B.2 follows from the definition of $G'$.
\end{proof}
%
%
%

\end{document}